\documentclass[11pt]{article}

\usepackage[T1]{fontenc}
\usepackage{lmodern}
\usepackage{textcomp}

\usepackage{amsmath}%
\usepackage{amssymb}%
\usepackage[table]{xcolor}%

\usepackage{todonotes}
\usepackage[in]{fullpage}%

\usepackage[amsmath,thmmarks]{ntheorem}%
\theoremseparator{.}%

\usepackage{titlesec}%
\titlelabel{\thetitle. }%
\usepackage{xcolor}%
\usepackage{mleftright}%
\usepackage{xspace}%
\usepackage{graphicx}
\usepackage{hyperref}%

\usepackage{hyperref}%
\hypersetup{%
      unicode,
      breaklinks,%
      colorlinks=true,%
      urlcolor=[rgb]{0.25,0.0,0.0},%
      linkcolor=[rgb]{0.5,0.0,0.0},%
      citecolor=[rgb]{0,0.2,0.445},%
      filecolor=[rgb]{0,0,0.4},
      anchorcolor=[rgb]={0.0,0.1,0.2}%
}
\usepackage[ocgcolorlinks]{ocgx2}

\theoremseparator{.}%

\theoremstyle{plain}%
\newtheorem{theorem}{Theorem}[section]

\newtheorem{lemma}[theorem]{Lemma}

\theoremstyle{plain}%
\theoremheaderfont{\sf} \theorembodyfont{\upshape}%
\newtheorem*{remark:unnumbered}[theorem]{Remark}%
\newtheorem{definition}[theorem]{Definition}

\newcommand{\myqedsymbol}{\rule{2mm}{2mm}}

\theoremheaderfont{\em}%
\theorembodyfont{\upshape}%
\theoremstyle{nonumberplain}%
\theoremseparator{}%
\theoremsymbol{\myqedsymbol}%
\newtheorem{proof}{Proof:}%

\providecommand{\emphind}[1]{}%
\renewcommand{\emphind}[1]{\emph{#1}\index{#1}}

\definecolor{blue25emph}{rgb}{0, 0, 11}

\providecommand{\emphic}[2]{}
\renewcommand{\emphic}[2]{\textcolor{blue25emph}{%
      \textbf{\emph{#1}}}\index{#2}}

\providecommand{\emphi}[1]{}%
\renewcommand{\emphi}[1]{\emphic{#1}{#1}}

\definecolor{almostblack}{rgb}{0, 0, 0.3}

\providecommand{\emphw}[1]{}%
\renewcommand{\emphw}[1]{{\textcolor{almostblack}{\emph{#1}}}}%

\providecommand{\emphOnly}[1]{}%
\renewcommand{\emphOnly}[1]{\emph{\textcolor{blue25}{\textbf{#1}}}}

\newcommand{\HLink}[2]{\hyperref[#2]{#1~\ref*{#2}}}
\newcommand{\HLinkSuffix}[3]{\hyperref[#2]{#1\ref*{#2}{#3}}}

\providecommand{\deflab}[1]{}
\renewcommand{\deflab}[1]{\label{def:#1}}

\providecommand{\eqlab}[1]{}%
\renewcommand{\eqlab}[1]{\label{equation:#1}}

\newcommand{\remove}[1]{}%

\usepackage[inline]{enumitem}

\newlist{compactenumA}{enumerate}{5}%
\setlist[compactenumA]{topsep=0pt,itemsep=-1ex,partopsep=1ex,parsep=1ex,%
   label=(\Alph*)}%

\newlist{compactenuma}{enumerate}{5}%
\setlist[compactenuma]{topsep=0pt,itemsep=-1ex,partopsep=1ex,parsep=1ex,%
   label=(\alph*)}%

\newlist{compactenumI}{enumerate}{5}%
\setlist[compactenumI]{topsep=0pt,itemsep=-1ex,partopsep=1ex,parsep=1ex,%
   label=(\Roman*)}%

\newlist{compactenumi}{enumerate}{5}%
\setlist[compactenumi]{topsep=0pt,itemsep=-1ex,partopsep=1ex,parsep=1ex,%
   label=(\roman*)}%

\newlist{compactitem}{itemize}{5}%
\setlist[compactitem]{topsep=0pt,itemsep=-1ex,partopsep=1ex,parsep=1ex,%
   label=\ensuremath{\bullet}}%

\numberwithin{figure}{section}%
\numberwithin{table}{section}%
\numberwithin{equation}{section}%

\usepackage{algorithm}
\usepackage{algpseudocode}
\usepackage{mathrsfs}
\usepackage{cleveref}
\usepackage{comment}
\usepackage{xcolor}
\usepackage{subcaption}
\usepackage{soul}
\usepackage{tikz}
\usetikzlibrary{arrows.meta}
\usepackage{authblk}

\newcommand{\DD}{\mathcal{D}}
\newcommand{\JJ}{\mathcal{J}}
\newcommand{\MM}{\mathcal{M}}
\newcommand{\PP}{\mathcal{P}}

\newcommand{\HH}{\mathcal{H}}
\newcommand{\LPT}{Alg}
\newcommand{\OPT}{Opt}

\newcommand{\ARI}[1]{\frac{T(\LPT_{#1})}{T(\OPT_{#1})}}
\newcommand{\Red}{R}
\newcommand{\II}{\mathcal{I}}

\DeclareMathOperator*{\argmin}{arg\,min}

\begin{document}

\title{Two Results on LPT: A Near-Linear Time Algorithm
and Parcel Delivery using Drones}

\author[1]{L. Sunil Chandran}
\author[1]{Rishikesh Gajjala}
\author[1,2]{Shravan Mehra}
\author[1]{Saladi Rahul}

\affil[1]{Indian Institute of Science, Bengaluru, India}
\affil[2]{University of Birmingham, UK}

\date{ }

\maketitle

\begin{abstract}
The focus of this paper is to increase our understanding of the {\em Longest Processing Time First (LPT)} heuristic. LPT is a classical heuristic for {the fundamental
problem} of uniform machine scheduling. For different machine speeds, LPT was first considered by Gonzalez et al. {\em(SIAM J. Comput. 6(1):155–166, 1977)}. Since then, extensive work has been done to improve the approximation factor of the LPT heuristic. However, all known implementations of the LPT heuristic take $O(mn)$ time, where $m$ is the number of machines and $n$ is the number of jobs. In this work, we come up with the first {\em near-linear time} implementation for LPT. Specifically, the running time is $O((n+m)(\log^2{m}+\log{n}))$. Somewhat surprisingly, the result is obtained by mapping the problem
to dynamic maintenance of lower envelope of lines, which has been well studied in the computational geometry community.

Our second contribution is to analyze the performance of LPT for the {\em Drones Warehouse Problem (DWP)}, which is a natural generalization of the uniform machine scheduling problem motivated by drone-based parcel delivery from a warehouse. In this problem, a warehouse has multiple drones and wants to deliver parcels to several customers. Each drone picks a parcel from the warehouse, delivers it, and returns  to the warehouse (where it can also get charged). The speeds and battery lives of the drones could be different, and due to the limited battery life, each drone has a bounded range in which it can deliver parcels. The goal is to assign parcels to the drones so that the time taken to deliver all the parcels is minimized. We prove that the natural approach of solving this problem via the LPT heuristic has an approximation factor of $\phi$, where $\phi \approx 1.62$ is the golden ratio.
\end{abstract}

\section{LPT heuristic for uniform scheduling}
Uniform machine scheduling with the minimum makespan objective is a {fundamental} problem. In this problem, we are given a set of $n$ jobs (not necessarily of the same size) and a set of $m$ machines (not necessarily of the same speeds). The goal is to schedule the $n$ jobs {on} $m$ machines so that the time required to execute the schedule (makespan) is minimised. This is an NP-hard problem even for two machines \cite{garey1979computers} of the same speed, but polynomial time approximation schemes (PTASs) are known \cite{HochbaumS87, HochbaumS88}.

A commonly studied heuristic for this problem is the {\em Longest Processing Time First} (LPT) heuristic. In the LPT heuristic, each job is assigned one by one, in non-increasing order of size, so that every job is assigned to a machine where it will be completed earliest (with ties being broken arbitrarily). Note that a machine might already have some jobs assigned to it and the execution of the current job happens only after finishing the already assigned jobs. 

More intricate algorithms were 
designed in the literature to get a good approximation factor for uniform scheduling. For example,
Horowitz and Sahni gave an exact dynamic programming algorithm which runs in exponential time \cite{HorowitzS76}. When there are only two machines, they could build upon this algorithm to obtain a Polynomial-time approximation scheme (PTAS). Later, Hochbaum and Shmoys gave a PTAS when all the machines had identical speeds~ \cite{HochbaumS87}. This was later extended to obtain a PTAS for the uniform scheduling problem (USP)~\cite{HochbaumS88}. 

However, the LPT algorithm 
remains popular in practice 
due to its simplicity and scalability
(compared to the PTAS-type results which
are relatively complicated and have expensive
running time). As a result,
there has been a long line of research on improving the approximation ratio of the LPT algorithm for USP. In two independent works, 
the ratio of the LPT algorithm
was improved by Dobson \cite{Dobson} to $\frac{19}{12}$ and by Friesen \cite{Friesen} to $\frac{5}{3}$. Kovacs \cite{Kovacs10} further improved the approximation factor of the LPT algorithm to $1.58$ and proved that the LPT algorithm cannot give an approximation factor better than $1.54$. 

\section{First result: A near-linear 
time implementation of LPT}
In spite of all the focus in the literature 
on adapting {LPT} to various settings
of machine scheduling and analyzing its approximation factor, to the best of our knowledge, there has
been no work on  fast implementation of LPT.
The known implementation of the LPT heuristic takes $O(mn)$ time (via the naive approach). In this work, we give the first near-linear time implementation of the LPT heuristic. 
\begin{theorem}\label{near_linear_thm}
There is an $O((n+m)(\log^2{m}+\log{n}))$ time 
implementation of the LPT heuristic.   
\end{theorem}

For a set of given lines {in 2-
D}, 
the {\em lower envelope} is the point-wise minimum of the lines {(a more formal
definition will follow later)}. In the dynamic maintenance of the lower envelope problem, in each step, a new line is added (or removed), as shown in \cref{lower_envelope_fig}, and one has to maintain the lower envelope with a small update time. This has been well-studied in the computational geometry community. We establish a connection from LPT to the dynamic maintenance of the lower envelope of lines to prove \cref{near_linear_thm} in \cref{near_linear_sec}.

 
\section{Second result: LPT for the
drones warehouse problem (DWP)}

Our second contribution is to analyze the
performance of LPT for the {\em Drones warehouse problem (DWP)} (formally defined in \cref{subsection:dwp}) which is a natural generalization 
of the uniform scheduling problem to drone-based parcel delivery from a warehouse.

{\em Vehicle routing} \cite{truck1,truck2}
is a classical problem in which parcel deliveries are done by a single truck or a collection of trucks. {Researchers have explored variations with different vehicle velocities \cite{vehicle-routing-different-speeds} or scenarios where each parcel can only be delivered by a specific subset of vehicles \cite{vehicle-routing-compatibility}.} With the advent of drones, a generalization of the vehicle routing problem has been studied in the literature in which a truck is carrying drones along with it. The drones pick up parcels from the truck, deliver the package and return to the truck (the truck might have now moved to a different location). This problem is more challenging than the traditional
vehicle routing problem \cite{DBLP:journals/mansci/CarlssonS18}. Several MIP (mixed integer programming) formulations and heuristics have been used to solve this problem \cite{truckuav1,truckuav2}. Theoretical guarantees have also been proved for this problem by Carlsson and Song using geometric methods  \cite{DBLP:journals/mansci/CarlssonS18}.

The transition towards using only drones like in DWP (instead of trucks with drones) is evident in the gradual shift within the research community,
as it is a more sustainable option for the future. Extensive efforts have been dedicated to developing algorithms for scheduling drones under various constraints \cite{UAVALGOBOOK}. 
Some MIP formulations were studied to minimize various objectives like the number of drones \cite{UAV1,UAV2} and bio-inspired algorithms were used to manage large fleets of drones \cite{UAV4}. As drones have a limited battery life, considerations for fuel stations were explored in \cite{UAV3}. Further, the drones might not all have the same features like speed and battery life and this was taken into account in \cite{UAV5}. For a comprehensive review of research in several related problems involving only drones, the reader can refer to the survey presented in \cite{UAV_survey}. There has also been a lot of work by the multi-agent community for scheduling \cite{AAMAS20}, pathfinding \cite{AAMAS22,AAMAS18} and coordinating drones \cite{AAMASA17, DBLP:conf/ijcai/Manoharan20, DBLP:conf/ijcai/BodinCQS18, DBLP:conf/ijcai/WuRC16}.


In this paper, 
we consider the problem where a warehouse wants to use {\em drones} to deliver a large number of parcels to customers around it. 
Due to limited battery life, each drone has a restricted
range around the warehouse in which it can deliver parcels. 
Also, depending on the manufacturer, the speed of each drone 
can vary. We will now formally define the DWP problem, state the result obtained by applying the LPT heuristic for DWP and provide a high-level overview of the analysis. We also provide a detailed literature review on the use of LPT for several other machine scheduling problems and connect it to our result on DWP in \cref{machine_scheduling_overview}.

\subsection{The drones warehouse problem (DWP)}\label{subsection:dwp}

For the sake of better readability, we deviate from the notation used in scheduling literature. The warehouse has a set of $m$ drones 
$\mathcal{D}=\{D_1,D_2\cdots D_m\}$ and a set of 
$n$ parcels $\mathcal{P}=\{P_1,P_2\cdots P_n\}$. The parcels in set $\mathcal{P}$ need to be delivered by drones in 
set $\mathcal{D}$. Each drone can pick up one parcel at a time from the warehouse, deliver it and return to the warehouse. 
For all $1\leq j\leq n$, let the distance at which parcel $P_j$ needs to be delivered be $\ell_j/2$. So each drone must travel a distance of $\ell_j$ in total to deliver the parcel $P_j$ and 
come back to the warehouse.

Additionally, the drones have a limited battery life which can be different for each drone. Let $d_i$ be the distance which drone $D_i$ can travel,
for all $1\leq  i \leq m$. Therefore, for a parcel $P_j$ to be delivered by a drone $D_i$, it must be the case that $\ell_j \leq d_i$. The speed at which the drone $D_i$ travels is $v_i$, for all $1\leq i\leq m$. After each delivery, the drone recharges its battery in the warehouse, and 
for the sake of simplicity we assume the time taken to recharge is negligible. Our goal is to assign each parcel to a drone such that the time taken to execute the schedule is minimised.


More precisely, we define a \textit{valid schedule} $f: \DD \rightarrow 2^{\PP}$ (power set of $\PP$) such that the following properties hold
\begin{itemize}
    \item Each parcel is assigned to exactly one drone, i.e., $f(D_i) \cap f(D_j) = \emptyset$ for all $i\neq j$ and $\bigcup_{D\in \DD} f(D) = \PP$.
    \item Each drone must be able to deliver the parcel assigned to it, i.e., $\ell_i \leq d_j$ for all $j\in [m]$ and $i: P_i \in f(D_j)$.
\end{itemize}
Our goal is to find a \textit{valid schedule} such that $T(f)$ is minimised where 
$$T(f) = \max\limits_{j\in[m]}\sum_{i: P_i \in f(D_j)} \frac{\ell_i}{v_j}$$
We assume that there is at least one drone capable of delivering all parcels. Otherwise, there would be no valid solution (this can be checked in linear time).


\subsection{Our results and techniques} 
We implement the LPT algorithm with additional battery life constraints to solve DWP in near-linear time. Our key contribution is to prove that this algorithm always returns a solution which has delivery time at most $\phi$ times the optimal solution, where $\phi \approx 1.62$ is the golden ratio. We summarize our results and compare them with previous work in \cref{Fig:1} (discussed in more detail in \cref{machine_scheduling_overview})

\begin{theorem}\label{thm:phi_approx}
There is a $\phi$-approximation algorithm for the DWP problem where $\phi = \frac{1 + \sqrt{5}}{2}$ is the golden ratio. The algorithm runs in $O((n+m)(\log^2{m}+\log{n}))$ time, where $m$ and $n$ are the number of drones and parcels, respectively. 
\end{theorem}
At a high level, our proof is inspired by the analysis {of the LPT algorithm} for the uniform scheduling problem (USP) \cite{Kovacs10}. However, the constraint of battery life makes our problem significantly more challenging. As the total distance travelled by the drones  is the same for any valid schedule, if some drone in the LPT algorithm travels more distance than its counterpart in the optimal assignment, then some other drone in the LPT algorithm must travel lesser distance than its counterpart in the optimal assignment as a compensation. However, if we assume that the approximation factor is greater than $\phi$, we can create instances for which such compensation does not occur, which leads to a contradiction. The proof requires ideas such as
(a) working with a minimal instance, (b) removing the parcel which is closest to the warehouse from the schedule, (c) classifying the parcels into three categories based on their distance and 
(d) truncating the parcel distances. We give the complete proof in~\cref{MLPT}.






\section{Related work on machine scheduling}
\label{machine_scheduling_overview}
We will now give a detailed literature review of the use of LPT for machine scheduling and connect our problem DWP with the other variants.

\begin{figure*}[t]
    \centering

\begin{tikzpicture}

\tikzstyle{node} = [rectangle, draw, rounded corners, minimum width=1.5cm, minimum height=1cm, align=center]
\tikzstyle{approx} = [font=\small, align=center]

\node[node] (ISP) at (0,0) {ISP};
\node[node] (USP) at (3,0) {USP};
\node[node] (DWP) at (6,0) {DWP};
\node[node] (USP-C) at (9,0) {USP-C};
\node[node] (UnrelSP) at (12,0) {UnrelSP};

\node[approx] at (0,-1.2) {$=4/3$ \\ \footnotesize{\cite{Graham1}}
};
\node[approx] at (3,-1.2) {$[1.54,1.58]$ \\ \footnotesize{\cite{Kovacs10}}
};
\node[approx] at (6,-1.2) {$[1.54,\phi]$ \\ 
\footnotesize{[This work]}
};
\node[approx] at (9,-1.2) {$\geq 2$ \\ \footnotesize{[Folklore]}
};
\node[approx] at (12,-1.2) {Not Applicable \\ 
};

\node at (1.5,0) {$\subset$};
\node at (4.5,0) {$\subset$};
\node at (7.5,0) {$\subset$};
\node at (10.5,0) {$\subset$};

\end{tikzpicture}

    \caption{Landscape of the approximation ratio of the LPT heuristic for machine scheduling problems and our 
    drone warehouse problem (DWP).
    The relation $A \subset B$ in the figure implies that $A$ is a
    special case of $B$. Therefore, the approximation
    factor increases from left to right in the figure. The interval $[a,b]$ means that the approximation ratio of the LPT heuristic is at least $a$ and at most $b$. As there is no total order among jobs in UnrelSP, LPT is not applicable for UnrelSP.    
   } 
    \label{Fig:1}
\end{figure*}
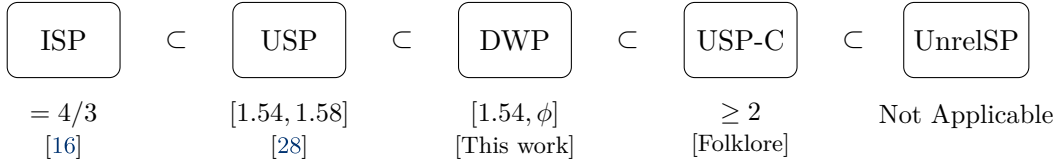

\subsection{{Uniform machines
scheduling problem} (USP)}
{In the {\em uniform 
machines scheduling problem},} let $\PP = \{P_1, P_2, \ldots, P_n\}$ be a set of $n$ jobs of size $\{\ell_1, \ell_2, \ldots, \ell_n\}$ respectively and $\DD = \{D_1, D_2, \ldots, D_m\}$ be a set of $m$ machines of speed $\{v_1, v_2, \ldots, v_m\}$ respectively. Our goal is to schedule the $n$ jobs to the $m$ machines so that the completion time of the schedule is minimised. We will now formally define the problem.

We define a \textit{schedule} as a function $f: \DD \rightarrow 2^{\PP}$ (power set of $\PP$ ) where $f(D_i)$ represents the set of all jobs assigned to machine $D_i$. We call a \textit{schedule} a \textit{valid schedule} if each job is assigned to exactly one machine, i.e., $f$ is a \textit{valid schedule} if and only if $f(D_i)\cap f(D_j) = \emptyset$ for all $i\neq j$ and $\bigcup_{D\in \DD} f(D) = \PP$. Each schedule $f$ has an associated completion time 
$$T(f) = \max\limits_{j\in[m]}\sum_{i: P_i \in f(D_j)} \frac{\ell_i}{v_j}$$
Our goal is to find a \textit{valid schedule} $f$ such that $T(f)$ is minimised.

\subsection{ISP $\subset$ USP $\subset$ DWP} 
The special case of USP when all the machines have equal speed is called the {\em identical-machines scheduling problem (ISP)}. Therefore, we denote
ISP $\subset$ USP, where the notation $A \subset B$ means that 
$A$ is a special case of $B$ (See \cref{Fig:1}) and therefore, a lower bound of $A$ is also a lower bound for $B$ and an upper bound of $B$ is also an upper bound for $A$. Graham \cite{Graham1,Graham2} proved that the approximation factor of the LPT algorithm is $\frac{4}{3}$ for ISP. For USP, after a series of works, Kovacs \cite{Kovacs10} proved that the approximation factor of the LPT algorithm is at most $1.58$ and at least $1.54$.  It is easy to see that USP is a special case of DWP (USP $\subset$ DWP) when all battery lives are large enough to deliver all parcels. So, the approximation factor of LPT for DWP can not be better than $1.54$. On the other hand, due to \cref{thm:phi_approx}, we get that the LPT heuristic is a $\phi$-approximation, which is one of the main contributions of this work. 
\subsection{DWP $\subset$ USP-C $\subset$ UnrelSP.} \label{2approx_proof}
A lot of work in the literature has been devoted to a generalization of USP,
namely uniform scheduling problems with {\em processing constraints} (USP-C)
\cite{LEUNG2008251}. We are given a set of jobs $\JJ = \{J_1, J_2, \ldots, J_n\}$ and a set of machines $\MM$, where each job $J_j$ has a set of machines $\MM_j \subseteq \MM$ to which they can be assigned to. Different structural 
restrictions of $\MM_j$ lead to different models in USP-C, like the inclusive processing set \cite{ou2008, DBLP:journals/eor/LiW10a}, nested processing set \cite{MURATORE201047, EPSTEIN2011586}, interval processing set \cite{SHABTAY2012405, KARHI2014155}, and tree-hierarchical processing set restrictions \cite{DBLP:journals/jors/LiL15, DBLP:journals/jors/LiL16}. The goal is to assign each job to a machine to minimize the completion time. 

DWP is also a special case of USP-C as each parcel can only be delivered by a subset of drones determined by its battery life. Therefore, DWP $\subset$ USP-C. We emphasise that these two classes are strictly different (and DWP is also different from all known variants of USP-C, including nested intervals). Consider two machines and two jobs. Let the size of $J_1, J_2$ be $10,10+\epsilon$ respectively and the speeds of $M_1,M_2$ be $10,10+\epsilon$ respectively for $\epsilon > 0$. Moreover, assume that the job $J_1$ can only be done by $M_2$, but $J_2$ can be done by both $M_1$ and $M_2$. The LPT heuristic would take $\frac{20+\epsilon}{10+\epsilon}$ time, while the optimum assignment would take only $\frac{10+\epsilon}{10}$ time. This gives us an approximation ratio of $2$ as $\epsilon$ approaches zero in this example {(whether this ratio is optimal or not for the LPT heuristic on USP-C is an open question)}. We also note that this is not a valid lower bound instance for DWP as any drone which can do a job at a distance of $10+\epsilon$ can also do the job at a distance of $10$. 



The unrelated Scheduling Problem (UnrelSP) is the scheduling problem in which the time taken by machine $D\in \DD$ to complete job $P\in \PP$ is determined by an arbitrary function $f: \DD\times \PP \rightarrow \mathbb{R}$. Note that USP-C $\subset$ UnrelSP and furthermore the LPT heuristic is not applicable for UnrelSP as there is no total order among the jobs to sort. This finishes the description of \cref{Fig:1}.

\subsection{Other algorithms and optimization measures.}
A PTAS for ISP was given by Hochbaum and Shmoys~\cite{HochbaumS87}. This was later extended to obtain a PTAS for USP~\cite{HochbaumS88}. Due to the seminal result of Lenstra, Shmoys and Tardos \cite{2approxgen}, there is an LP-based 2-approximation algorithm for UnrelSP. This is also the best known approximation algorithm for USP-C. For a special case of USP-C (including DWP) with nested intervals, there is $4/3$ approximation algorithm \cite{DBLP:journals/eor/LeungN17}. For more results on USP-C, we refer the reader to the latest survey \cite{LEUNG20161}. We emphasise that despite knowing PTASs and other algorithms with a better approximation ratio, the LPT heuristic remains popular in practice 
due to its simplicity and scalability. This motivated researchers from the algorithms and operations research community to extensively study it as described in \cref{Fig:1}. 


Instead of optimizing the time taken to complete all jobs (makespan), other objectives like the average completion time \cite{DBLP:journals/cacm/BrunoCS74}, weighted-average completion time \cite{HorowitzS76} and monotonicity and truthfulness have also been studied \cite{mon1, mon4, mon3, mon2}.

\section{Near linear time implementation}\label{near_linear_sec}

\subsection{Longest processing time first (LPT)}
A classical approach for the Uniform Scheduling problem is using a greedy algorithm called LPT scheduling which gives a $1.58$-approximate solution \cite{Kovacs10}. First, we sort the jobs $\PP$ in decreasing order of their size 
{and let $P_1,P_2,\ldots, P_n$ be 
the sorted sequence}. Then we initialise $T_j$, the time taken by the $j$th machine to be zero for all $j \in [m]$. Now we assign the jobs sequentially from $P_1$ to $P_n$. We assign job $P_i$ to a machine $D_j$ for which the value of $T_j + (\ell_i/v_j)$ is minimum and we update the value $T_j$ to $T_j + (\ell_i/v_j)$ for this specific $j$ (see \cref{alg: LPT_slow}).

\begin{algorithm}
\caption{LPT algorithm in $O(nm)$ time}
\label{alg: LPT_slow}
\begin{algorithmic}
\State \textbf{Input: }List of jobs $\PP$ and machines $\DD$.
\State \textbf{Output: }Time required for LPT scheduling.
\State Sort $\PP$ in non-increasing order of size.
\State Initialise $T_j = 0$ for all $j \in [m]$
\For{$i$ in $\{1\ldots n\}$}
    \State $\alpha = \argmin_{j\in[m]} \left(T_j + \frac{\ell_i}{v_j}\right)$
    \State $T_{\alpha} \leftarrow T_{\alpha} + \frac{\ell_i}{v_{\alpha}}$
\EndFor
\State \textbf{return} $\max\limits_{j \in [m]} T_j$
\end{algorithmic}
\end{algorithm}

As $\argmin_{j\in[m]} T_j + \frac{\ell_i}{v_j}$ can be found in $O(m)$ time, we can implement the above algorithm to run in $O(nm)$ time. To improve the run time of this algorithm, we implement a faster way to find $\argmin_{j\in[m]} T_j + \frac{\ell_i}{v_j}$.

\subsection{Dynamic Lower Envelope}

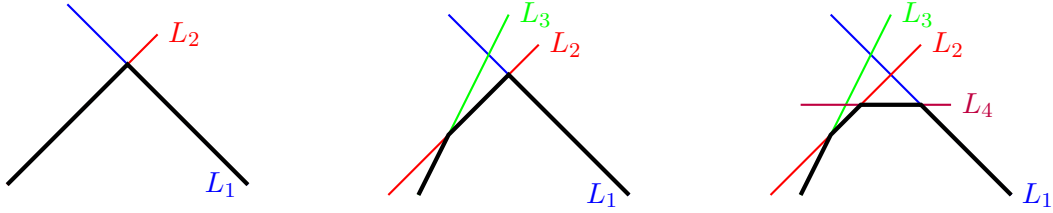
\begin{figure}
\begin{minipage}{0.3\textwidth}

\begin{tikzpicture}[scale=0.4]

\draw[thick, blue] (0, 4) -- (6, -2) node[anchor=east] {$L_1$};
\draw[thick, red] (-2, -2) -- (3, 3) node[anchor=west] {$L_2 $};

\draw[ultra thick, black] (-2, -2) -- (2, 2) -- (6, -2);



\end{tikzpicture}
\end{minipage}
\begin{minipage}{0.3\textwidth}

\begin{tikzpicture}[scale=0.4]

\draw[thick, blue] (0, 4) -- (6, -2) node[anchor=east] {$L_1$};
\draw[thick, red] (-2, -2) -- (3, 3) node[anchor=west] {$L_2$};
\draw[thick, green] (-1, -2) -- (2, 4) node[anchor=west] {$L_3$};

\draw[ultra thick, black] (-1, -2) -- (0, 0) -- (2, 2) -- (6, -2);



\end{tikzpicture}
\end{minipage}
\begin{minipage}{0.3\textwidth}

\begin{tikzpicture}[scale=0.4]

\draw[thick, blue] (0, 4) -- (6, -2) node[anchor=west] {$L_1$};
\draw[thick, red] (-2, -2) -- (3, 3) node[anchor=west] {$L_2$};
\draw[thick, green] (-1, -2) -- (2, 4) node[anchor=west] {$L_3$};
\draw[thick, purple] (-1, 1) -- (4, 1) node[anchor=west] {$L_4$};

\draw[ultra thick, black] (-1, -2) -- (0, 0) -- (1, 1) -- (3,1) -- (6, -2);



\end{tikzpicture}
\end{minipage}
   \caption{Dynamic lower envelope of lines $L_1: y=-x+4$, $L_2: y = x $, $L_3: y = 2x$, $L_4: y = 1$}
   \label{lower_envelope_fig}
\end{figure}

One can visualise the step in which $\argmin_{j\in[m]} T_j + \frac{\ell_i}{v_j}$ is computed in the following way: Consider the $m$ linear functions $h_1, h_2, \ldots, h_m$ where $h_j(x) =\dfrac{1}{v_j}\cdot x+T_j$. We need to find the index of $j \in [m]$ for which $h_j(x) =\dfrac{1}{v_j}\cdot x+T_j$ is minimum at $x=l_i$. This is exactly the same as finding the line in the lower envelope for the $h_1, h_2 \cdots h_m$ at $x=l_i$. Our idea now is to maintain a dynamic lower envelope data structure capable of inserting and deleting functions. We will now describe this formally.

Let $\HH$ be a set of functions where $h\in \HH$ is of the form $h: \mathbb{R}\rightarrow \mathbb{R}$. Then the lower envelope is the function $h_{min}:\mathbb{R}\rightarrow \mathbb{R}$ such that $h_{min}(x) = \min_{h\in \HH} h(x)$. We are interested in the case when these functions correspond to straight lines, i.e., they are of the form $h_i(x) = m_ix + c_i$. More particularly, we are interested in the problem of dynamically maintaining the lower envelope problem of lines, i.e., in each step, a new line is added (or removed), as shown in \cref{lower_envelope_fig}, and one has to maintain the lower envelope with a small update time. This has been well-studied in the computational geometry community. From 
\cite{DBLP:journals/jcss/OvermarsL81}, there is a data structure to maintain a dynamic lower envelope of lines (further extended to dynamic lower envelope of pseudo lines in \cite{DBLP:conf/compgeom/AgarwalCHM19}):
\begin{itemize}
    \item $\HH.Insert(h)$: Adds the linear function $h$ to set $\HH$ in $O(\log^2{|\HH|})$ time.
    \item $\HH.Delete(h)$: Removes the function $h$ from set $\HH$ in $O(\log^2{|\HH|})$ time.
    \item $\HH.LowerEnvelope(x)$: Returns $h^* = \argmin_{h\in \HH} h(x)$ in $O(\log{|\HH|})$ time.
\end{itemize}

\begin{algorithm}
\caption{LPT algorithm in $O((n+m)(\log^2{m}+\log{n})$ time}
\label{alg: LPT_fast}
\begin{algorithmic}
\State \textbf{Input: }List of jobs $\PP$ and machines $\DD$.
\State \textbf{Output: }Time required for LPT scheduling.
\State Sort $\PP$ in non-increasing order of size.
\State Initialise lower envelope data structure $\HH$
\For{$j$ in $\{1\ldots m\}$}
    \State $h_j(x) = \frac{1}{v_j}x$
    \State $\HH.Insert(h_j)$
\EndFor
\For{$i$ in $\{1\ldots n\}$}
    \State $h_j = \HH.LowerEnvelope(\ell_i)$
    \State Let $h_j(x) = \frac{1}{v_j}x + T_j$
    \State Set $h_j'(x) = \frac{1}{v_j}x + T_j + \frac{\ell_i}{v_j}$
    \State $\HH.Delete(h_j)$
    \State $\HH.Insert(h_j')$
\EndFor
\State \Return $\max\limits_{j\in [m]} h_j(0)$
\end{algorithmic}
\end{algorithm}

Using this, we can implement the LPT algorithm in $O((n+m))\log^2{m})$ time (\cref{alg: LPT_fast}). We initialise the lower envelope data structure $\HH$ and store lines $h_j(x) = \frac{1}{v_j}x + 0$ (initially $T_j = 0$) for all $j\in [m]$. To assign parcel $P_i$ to drone $D_j$, we remove the line $h_j(x) = \frac{1}{v_j}x + T_j$ from $\HH$ and replace it with $h_j'(x) = \frac{1}{v_j}x+ T_j + \frac{\ell_i}{v_j}$. To find out which drone $P_i$ is assigned to, we have to find $\argmin_{j\in[m]} T_j + \frac{\ell_i}{v_j}$, which is the same as querying $\HH.LowerEnvelope(\ell_i)$.

Observe that sorting $\PP$ takes $O(n\log{n})$ time. We have called $\HH.Insert(\cdot)$ $O(n+m)$ times and $\HH.Delete(\cdot)$ $O(n)$ times. Therefore, the runtime of the algorithm is $O((n+m)\log^2{m} + n\log{n})$ or $O((n+m)(\log^2{m} + \log{n}))$



\section{$\phi$-approximation for the Drone Warehouse Problem}\label{MLPT}

In this section, we will prove \cref{thm:phi_approx}.

\subsection{Algorithm}
We use an algorithm similar to LPT but modify it slightly so that the battery constraints are respected.
First, we sort the parcels $\PP$ in non-increasing order of their distance. Then, we initiate $T_j$, the time taken by the $j$th drone to be zero for all $j\in [m]$. We now assign the parcels sequentially from $P_1$ to $P_n$. We assign parcel $P_i$ to a drone $D_j$ for which $\ell_i \leq d_j$ and the value of $T_j + (\ell_i/v_j)$ is minimum. We update the value $T_j$ to $T_j+
(\ell_i/v_j)$ (see ~\cref{alg: MLPT}). 
It is easy to see that the running time of the LPT algorithm is $O(mn + n\log{n})$

\subsection{Implementation}
Let us sort all the drones in decreasing order of their battery life. Observe that if drone $D_j$ is capable of delivering parcel $P_i$, then all drones $D_{j'}, j' \leq j$ are capable of delivering parcel $P_i$. Therefore, we can represent all the drones capable of delivering $P_i$ by a pointer $ptr$ so that all drones $D_j, j\in \{1, \ldots, ptr\}$ can deliver parcel $P_i$. Also, as the parcels are sorted in decreasing order of distance, the value of $ptr$ will only increase in each iteration (as any drone capable of delivering $P_i$ can also deliver $P_{i'}, i'>i$). We again use the lower envelope data structure to implement LPT. (see \cref{alg: MLPT}).

\begin{algorithm} [H]
\caption{LPT algorithm for DWP}
\label{alg: MLPT}
\begin{algorithmic}
\State \textbf{Input: }List of drones $\DD$ and parcels $\PP$.
\State \textbf{Output: }Minimum time required to deliver all parcels.
\State Sort $\PP$ in non-increasing order of distance.
\State Sort $\DD$ in non-increasing order of battery life
\State Initialise lower envelope data structure $\HH$
\State $ptr \leftarrow 0$
\For{$i$ in $\{1\ldots n\}$}
    \While{$ptr < n$ and $d_{ptr+1} \geq \ell_i$}
    \State $ptr \leftarrow ptr + 1$
    \State $h_{ptr}(x) = \frac{1}{v_{ptr}}x$
    \State $\HH.Insert(h_{ptr})$
    \EndWhile
    \State $h_j = \HH.LowerEnvelope(\ell_i)$
    \State Let $h_j(x) = \frac{1}{v_j}x + T_j$
    \State Set $h_j'(x) = \frac{1}{v_j}x + T_j+\frac{\ell_i}{v_j}$
    \State $\HH.Remove(h_j)$
    \State $\HH.Insert(h_j')$
\EndFor
\State \Return $\max\limits_{j\in [m]} h_j(0)$
\end{algorithmic}
\end{algorithm}

Observe, that sorting $\PP$ and $\DD$ takes $O(n\log{n} + m\log{m})$ time. We call $\HH.Insert(\cdot)$ $O(n+m)$ times and $\HH.Delete(\cdot)$ $O(n)$ times. As these function calls only use $O(\log^2{m})$ time, the above algorithm runs in $O((n+m)(\log^2{m}+ log{n}))$ time.

\subsection{Proof for $\phi$-approximation}


\paragraph{Simplifying steps.}
Assume that the parcels are sorted in decreasing order of distance, i.e., if $i < j$ then $\ell_i \geq \ell_j$ for all $i,j \in [n]$. Let $\LPT_I$ represent the \textit{valid schedule} obtained from the LPT-algorithm and let $\OPT_I$ represent some fixed optimal \textit{valid schedule} on instance $I = \{\DD, \PP\}$. We show that $\ARI{I} \leq \phi$. 
We will start by 
performing some simplifying steps on $I$.
\begin{lemma}\label{simplify}
For any instance $I=\{\mathcal{D},\mathcal{P}\}$, we
can assume that $\ell_n=1$, 
$T(Opt_{I})=1$ and no drone 
has battery life less than $\ell_n$.
\end{lemma}
\begin{proof}
Observe that scaling all values $\{\ell_i\}_{i\in [n]}, \{d_j\}_{j\in [m]}$ by some constant $\alpha$ scales $T(\LPT_I)$ and $T(\OPT_I)$ by $\alpha$. Similarly, scaling all values $\{v_j\}_{j \in [m]}$ by some constant $\beta$ scales $T(\LPT_I)$ and $T(\OPT_I)$ by $\beta^{-1}$. However, this procedure does not affect the value of the approximation factor $\ARI{I}$. Therefore, we can choose values $\alpha, \beta$ such that $\ell_n = 1$ and $T(\OPT_I) = 1$  (choosing $\alpha = \ell_n^{-1}, \beta = \ell_n^{-1}T(\OPT_I)$ gives us the desired result). Also, as $P_n$ is the smallest job, we can remove all drones which do not have enough battery life to deliver it as such drones would be empty in any schedule. 
\end{proof}
\subsubsection{Idea-1: Working with 
minimal instances.} Our goal is to prove that 
$T(Alg_I) \leq \phi$ for all instances $I$. Towards a contradiction, assume that there exists an instance $I$ for which $\ARI{I} > \phi$. Among all such contradicting  instances, let $\II$ be a contradicting instance which has the {\em minimum} number of parcels. 
We will sometimes drop the subscripts in $\LPT_\II$ and $\OPT_\II$ and write them as $\LPT$ and $\OPT$, respectively, for simplicity.


\subsubsection{Idea-2: A schedule 
without the last parcel.} As $\II = \{\DD, \PP\}$ is a contradicting instance with the least number of parcels, it means that for any other instance $\II'$ with fewer parcels that $\II$ is not a contradicting instance. 
In particular, this is true for $\II' = \{\DD, \PP\setminus\{P_n\}\}$. Intuitively, this implies $T(\LPT_{\II'}) \leq \phi$ and $T(\LPT) > \phi$,
and adding parcel $P_n$ causes the increase in time. We will now prove this rigorously.
 
\begin{definition}\label{LPT0}
    Let $\LPT_0$ be the schedule obtained by removing parcel $P_n$ from the schedule $\LPT$, i.e., $\LPT_0$ is a schedule such that $\LPT_0(D_i) = \LPT(D_i)\setminus \{P_n\}$ for all $i \in [m]$.     
\end{definition}

\begin{definition}\label{Li}
    Let $L_j(f)$ represent the total distance travelled by drone $D_j$ in schedule $f$. Then, $$L_j(f) = \sum_{i: P_i \in f(D_j)} \ell_i.$$
\end{definition}

\begin{definition}
    The completion time of a drone in schedule $f$ is the time taken by the drone to deliver all parcels assigned to it by the schedule $f$.
\end{definition}

\begin{lemma}\label{best_lemma}
    $\ell_n+ L_i(\LPT_0)=1 + L_i(\LPT_0) > \phi v_i$, for all $i \in [m]$.
\end{lemma}
\begin{proof}
    We claim that $T(\LPT_0) \leq \phi$. Suppose not. Then consider the instance $\II' = \{\DD, \PP\setminus\{P_n\}\}$ obtained from the minimal instance $\II = \{\DD, \PP\}$. Observe that the schedule $\LPT_{\II'}$ 
    and  schedule $\LPT_0$ are the same (as the assignment of the first $n-1$ parcels is independent of the $n^{th}$ parcel). Therefore, $T(\LPT_{\II'}) > \phi$. Also, observe that $T(\OPT_{\II'}) \leq T(\OPT_\II) = 1$. This implies that $\ARI{\II'} > \phi$ which contradicts the fact that $\II$ is an instance with the least number of parcels such that $\ARI{\II} > \phi$.
    
    Now as $T(\LPT) > \phi$ and $T(\LPT_0) \leq \phi$, this means that only the drone which delivers $P_n$ has completion time greater than $\phi$. Also as the LPT algorithm assigns parcels to drones which have the least completion time, this implies that assigning $P_n$ to any drone $D_i$ in $\LPT_0$ would result in its completion time being greater than $\phi$. Note that $P_n$ can be assigned to any drone as all drones have battery life at least $\ell_n$ (from~\cref{simplify}).

    Therefore, $(L_i(\LPT_0) + \ell_n)/v_i > \phi$. As $\ell_n = 1$, it follows that  $L_i(\LPT_0) + 1 > \phi v_i$.
\end{proof}

\subsubsection{Idea-3: Classification of parcels.}
We classify parcels for which the drone travels a distance in the range $[1, \phi)$ as {\em small} jobs, $[\phi, 2)$ as {\em medium} jobs and $[2, \infty)$ as {\em large} jobs. We first define a rounding function $\Red(x)$ such that 
$$\Red(x) =
 \begin{cases} 1 & \text{ if } x \in [1, \phi)\\
                               1.5& \text{ if } x \in [\phi, 2)\\
                               \lfloor x \rfloor & \text{ if } x \in [2, \infty)
        \end{cases}$$
Note that $\Red$ is a function such that 
$$x \geq \Red(x) \geq x/\phi 
\text{ for all } x \geq 1$$

Define $L'_i(f) = \sum_{j: P_j \in f(D_i)} \Red(\ell_j)$.
The motivation for defining $R(x)$ is that the value of $L'_i(f)$ can only be integer multiples of 0.5 whereas $L_i(f)$ can take any arbitrary value depending on the input.

\subsubsection{Idea-4: Discretizing Distances.}
Ideally, we would like to show that each drone in $\LPT_0$ travels more distance than its counterpart in $\OPT$. This would show that the total distance travelled by drones in $\LPT_0$ is greater than that of drones in $\OPT$ which is a contradiction as fewer parcels have been delivered in $\LPT_0$ than in $\OPT$. Unfortunately, it turns out that $L_i(\LPT_0)$ can be lesser than $L_i(\OPT)$. Instead, we show that
$L'_i(\OPT) \leq L'_i(\LPT_0)$ for all $i \in [m]$ and arrive at a similar contradiction by analogous argument.

\begin{lemma}\label{phi}
    $L'_i(\LPT_0) + \phi - 1 > v_i$ for all $i \in [m]$
\end{lemma}
\begin{proof}
    From~\cref{best_lemma}, we get that
    \begin{align*} 
    v_i &< \frac{1 + L_i(\LPT_0)}{\phi} 
    \leq \frac{1}{\phi} + L'_i(\LPT_0)
    = \phi - 1 + L'_i(\LPT_0)
    \end{align*}
The second inequality and the third equality follow from $x/\phi \leq R(x)$ and $\phi$ being  the golden ratio, respectively.
\end{proof}

\begin{lemma}\label{half}
    If a drone $D_i$ has a medium job assigned to it, then $L'_i(\LPT_0) + 0.5 > v_i$
\end{lemma}
\begin{proof} 
    Let a medium-sized job of size $y$ be assigned to drone $D_i$. Then,
    \begin{align*}
        L'_i(\LPT_0) &= R(y) + \sum_{z\in \LPT_0(D_i)\setminus\{y\}} R(z)
        \intertext{Since $y\in[\phi, 2)$, we get $\Red(y) = 1.5 > y - 0.5$, and hence,}
        L'_i(\LPT_0)&> y -0.5 + \sum_{z\in \LPT_0(D_i)\setminus\{y\}} R(z) \geq y - 0.5 + \frac{L_i(\LPT) - y}{\phi}
        \intertext{Using~\cref{best_lemma} and $y \geq \phi$, we get}
        L'_i(\LPT_0)&> y - 0.5 + \frac{\phi v_i - 1 - y}{\phi} = y(1- \frac{1}{\phi}) - 0.5 - \frac{1}{\phi} + v_i
        \\ &
        \geq \phi - 1 - \frac{1}{\phi} + v_i - 0.5 = v_i - 0.5
    \end{align*}
\end{proof}

\begin{lemma}\label{main}
    $L'_i(\OPT) \leq L'_i(\LPT_0)$ for all $i \in [m]$
\end{lemma}

\begin{proof}
    Towards a contradiction, let us assume that $L'_j(\OPT) > L'_j(\LPT_0)$ for some drone $D_j$. Observe that $L'_j(\OPT) \leq L_j(\OPT) \leq v_j$ as $T(\OPT) = 1$. Using this and ~\cref{phi} we get,
    \begin{align*}
        L'_j(\LPT_0) < L'_j(\OPT) &\leq v_j < \phi - 1 + L'_j(\LPT_0)   
    \end{align*}
    As $L'_j(\LPT_0)$ and $L'_j(\OPT)$ are both integer multiples of $0.5$, the only value which satisfies the inequality is, $L'_j(\OPT) = L'_j(\LPT_0) + 0.5$. Let $L'_j(\LPT_0) = k/2$ where $k$ is an integer.

    \noindent\textbf{Case 1:} $L'_j(\LPT_0) = k/2$ is not an integer.\\
    This can only happen if $D_j$ contains at least one medium job in $\LPT_0$ (as small and large sized jobs have integral values and cannot add to give a non-integer). But then by~\cref{half}, we get $L'_j(\OPT) = L'_j(\LPT_0) + 0.5 > v_j$, which is a contradiction as $L'_j(\OPT) \leq v_j$.

    \noindent\textbf{Case 2:} $L'_j(\LPT_0) = k/2$ is an integer.\\
    This means that $L'_j(\OPT) = (k+1)/2$ is not an integer and hence, $D_j$ has at least one medium job assigned to it in $\OPT$. Therefore, by definition of the $\Red(\cdot)$ function, we get 
    $L_j(\OPT) \geq L'_j(\OPT) + \phi - 1.5 = L'_j(\LPT_0) + \phi - 1$. Using~\cref{phi}, we get $L_j(\OPT) > v_j$, which is a contradiction as $L_j(\OPT) \leq v_j$.
\end{proof}

Using~\cref{main}, we get that $$\sum_{i = 1}^m L'_i(\OPT) \leq \sum_{i = 1}^m L'_i(\LPT_0).$$ Observe that $\sum_{i=1}^m L'_i(\OPT) = \sum_1^{n} \Red(\ell_i)$ and $\sum_{i=1}^m L'_i(\LPT_0) = \sum_1^{n-1} \Red(\ell_i)$. Substituting this, we get
$$\sum_1^{n} \Red(\ell_i) \leq \sum_1^{n-1} \Red(\ell_i),$$
which is a contradiction. Therefore, our initial assumption must be wrong, implying that $T(Alg) \leq \phi$.

\subsection{How much can the above analysis of LPT be improved?}
Recall that the special case of DWP when all drones have the same battery life ($d_i=d_j \geq \max_k (\ell_k)$ for all $i,j \in [m]$) is equivalent to USP. It is known that LPT cannot give an approximation better than $1.54$ for USP \cite{Kovacs10}. Therefore, the LPT algorithm can also not give an approximation ratio better than $1.54$ for DWP.  

\section{Future work}
A couple of immediate open problems are 
the following: 
\begin{enumerate}
\item Can the implementation of the 
    LPT heuristic be done in optimal time, i.e. $O((n+m)\cdot(\log m +\log n)$ time? We 
    believe that it should be possible,
    since the jobs are known upfront, i.e.,
    it is actually an offline problem.
    \item Can the approximation ratio of the LPT heuristic be improved for DWP from $\phi$?   
\end{enumerate}
In general, delivering parcels from 
the warehouse using drones is a rich 
source of scheduling and vehicle routing
problems. We 
mention two general directions:
\begin{enumerate}
    \item  As a concrete 
    setting, consider a warehouse which
    has a truck and a drone, both of 
    which operate independently. As in the paper, the goal is to assign 
    parcels to the truck and the drone
    so that time taken to deliver is minimized. 
    The parcels will be delivered by the drone using 
    the same model as in this paper, whereas the 
    truck will deliver using the traditional technique.
    A generalized version 
    of the problem would involve 
    multiple trucks and drones.
    \item Some companies might have 
    multiple warehouses and as such, 
    a parcel can be delivered by any
    of the warehouses. As a concrete
    setting, consider the generalization of  the DWP 
    studied in this paper to the setting
    where there are {\em two} warehouses
    at different locations.
    The goal is to perform a two-level
    partition: first partition the 
    parcels among the warehouses and then 
    partition them among the drones. 
    The goal is to deliver all the parcels
    as quickly as possible.
\end{enumerate}

\bibliographystyle{plainurl}
\bibliography{references}

\end{document}